\documentclass[american,aps,pra,reprint,superscriptaddress,showpacs]{revtex4-1}
\usepackage[T1]{fontenc}
\usepackage[latin9]{inputenc}
\usepackage{color}
\usepackage{amsthm}
\usepackage{amsmath, amssymb}
\usepackage{color}
\usepackage{bm}
\usepackage{bbm}
\usepackage{empheq}
\usepackage{pxfonts,txfonts}
\usepackage{tgpagella}
\usepackage[T1]{fontenc}
\usepackage{graphicx}
\usepackage{txfonts}
\usepackage{braket}
\usepackage[unicode=true,pdfusetitle,
 bookmarks=true,bookmarksnumbered=false,bookmarksopen=false,
 breaklinks=false,pdfborder={0 0 0},backref=false,colorlinks=true]
 {hyperref}
\hypersetup{
 linkcolor=urlcolor,citecolor=urlcolor,urlcolor=urlcolor}

\makeatletter
\@ifundefined{textcolor}{}
{%
 \definecolor{BLACK}{gray}{0}
 \definecolor{WHITE}{gray}{1}
 \definecolor{RED}{rgb}{1,0,0}
 \definecolor{GREEN}{rgb}{0,1,0}
 \definecolor{BLUE}{rgb}{0,0,1}
 \definecolor{CYAN}{cmyk}{1,0,0,0}
 \definecolor{MAGENTA}{cmyk}{0,1,0,0}
 \definecolor{YELLOW}{cmyk}{0,0,1,0}
 }
\theoremstyle{plain}
\newtheorem{thm}{\protect\theoremname}
\theoremstyle{plain}

\ifx\proof\undefined
\newenvironment{proof}[1][\protect\proofname]{\par
\normalfont\topsep6\p@\@plus6\p@\relax
\trivlist
\itemindent\parindent
\item[\hskip\labelsep
\scshape
#1]\ignorespaces
}{%
\endtrivlist\@endpefalse
}
\providecommand{\proofname}{Proof}
\fi
\theoremstyle{plain}

\theoremstyle{plain}
\newtheorem{cor}[thm]{\protect\corrolaryname}
\theoremstyle{plain}
\newtheorem{con}[thm]{\protect\conjecturename}

\definecolor{urlcolor}{rgb}{0,0,0.7}
\newcommand{\id}{\mathbbm{1}}
\newcommand{\CC}{\mathbbm{C}}
\newcommand\numberthis{\addtocounter{equation}{1}\tag{\theequation}}
\newcommand{\ran}{\rangle}
\newcommand{\lan}{\langle}
\DeclareMathOperator{\tr}{Tr}
\DeclareMathOperator{\diag}{diag}

\makeatother

\providecommand{\lemmaname}{Lemma}
\providecommand{\propositionname}{Proposition}
\providecommand{\theoremname}{Theorem}
\providecommand{\corrolaryname}{Corrolary}
\providecommand{\conjecturename}{Conjecture}

\begin{document}

\title{Sufficient separability criteria and linear maps}

\author{Maciej Lewenstein}
\email{maciej.lewenstein@icfo.es}
\affiliation{ICFO -- Institut de Ci\`encies Fot\`oniques, The Barcelona Institute
of Science and Technology, 08860 Castelldefels, Spain}
\affiliation{ICREA -- Instituci\'o Catalana de Recerca i Estudis Avan\c cats, Lluis
Companys 23, 08010 Barcelona, Spain}

\author{Remigiusz Augusiak}
\affiliation{ICFO -- Institut de Ci\`encies Fot\`oniques, The Barcelona Institute
of Science and Technology, 08860 Castelldefels, Spain}
\affiliation{Center for Theoretical Physics, Polish Academy of Sciences, Al. Lotnik\'ow 32/46, 02-668 Warsaw, Poland}

\author{Dariusz Chru\'sci\'nski}
\affiliation{Institute of Physics, Faculty of Physics, Astronomy and Informatics,   Nicolaus Copernicus University, Grudziadzka 5/7, 87--100 Torun, Poland}

\author{Swapan Rana}
\affiliation{ICFO -- Institut de Ci\`encies Fot\`oniques, The Barcelona Institute
of Science and Technology, 08860 Castelldefels, Spain}

\author{Jan Samsonowicz}
\affiliation{Faculty of Mathematics and Information Science,
Warsaw University of Technology, Pl. Politechniki 1, 00-61
Warszawa, Poland}

\begin{abstract}
We study families of positive and completely positive maps acting
on  a bipartite system  $\CC^M\otimes \CC^N$ (with $M\leq
N$). The maps have  a property that when applied to any state (of a given entanglement class) 
they result in a separable state, or more generally a state of another  certain entanglement
class (e.g., Schmidt number $\leq k$).  This allows us to derive  useful families of sufficient
separability criteria. Explicit examples of such
criteria have been constructed for arbitrary $M,N$, with a special emphasis
on $M=2$.  Our results can be viewed as
generalizations of the known facts that in the sufficiently close
vicinity of the completely depolarized state (the
normalized identity matrix), all states are separable (belong to
"weakly" entangled classes). Alternatively, some of our results
can be viewed as an entanglement classification for a certain family of states,
corresponding to mixtures of the completely polarized state with
pure state projectors, partially transposed and locally
transformed pure state projectors.
\end{abstract}


\maketitle

\section{Introduction}

Entanglement is a property of density operators identified with
quantum mechanical states of a  multipartite system. Such
operators act on the corresponding Hilbert space, which has a {\it
fixed} tensor structure, corresponding to different parties
(Alice, Bob, Charlie,...). Notably, entanglement is the most
fundamental resource for quantum information processing
\cite{4Horodeckis.RMP.2009}. For this reasons characterization of
entanglement and related quantum correlations is one of the most
important tasks in quantum information science. Yet, it is a very
challenging task, since even for bipartite systems the problem of
determining if a given state is entangled or not is NP-hard \cite{Gurvits.JCSS.2003} 
from the computational point of view. Although in low
dimensions there exist efficient numerical codes checking
entanglement using semi-definite programming \cite{Doherty+2.PRL.2002},
this problem becomes untractable when the dimension of the
underlying tensor product Hilbert space grows. For this reason it
is extremely important to develop and have at our disposal
operational entanglement criteria that can be used in experiments.

So far, most of the known separability criteria are necessary
(that is sufficient for detecting entanglement), starting with the
prominent positive partial transpose criterion, proposed by
Peres \cite{Peres.PRL.1996}, and also proven to be sufficient for 
two-qubits and one qubit-qutrit by Horodeckis
\cite{3Horodecki.PLA.1996}. In fact transposition $T$ is a
paradigmatic example of a {\it positive map}, i.e., a map that maps
states onto states. Transposition is not, however, completely
positive, i.e., its extensions by identity map to larger tensor
spaces do not correspond to positive maps. This is actually a
general fact, proven by Horodeckis (cf.\cite{4Horodeckis.RMP.2009}): 
any positive map acting exclusively
on a space of Alice's (Bob's) operators, transforms a separable state
into a positive definite state. Conversely, a state is separable
if it remains positive definite under action of all positive map
on Alice's (Bob's) side. Several examples of positive maps have been
proposed and studied in the context of entanglement: cf.  the
celebrated reduction map \cite{M+PHorodecki.PRA.1999,Cerf+2.PRA.1999}, 
Breuer-Hall map
\cite{Breuer.PRL.2006,Hall.JPA.2006}, or Choi map \cite{Choi.LAA.1975}.
Note that in contrast to the Breuer-Hall and Choi maps, the reduction map
is decomposable, i.e., it does not allow to detect entanglement of
states with positive partial transpose. Also, note that positive
maps are not "physical" and cannot be realized directly in
experiments; nevertheless, it is possible to realize them
indirectly, using the, so called, structural approximations to
positive maps
\cite{Fiurasek.PRA.2002,Horodecki+Ekert.PRL.2002,Korbicz+4.PRA.2008,Czekaj,Chruscinski.JPA.2014}, 
or multi-copy entanglement witnesses \cite{Carteret,HAD}.

Another common approach to assure separability and detect
entanglement is based on the celebrated  entanglement  witnesses (cf.
\cite{4Horodeckis.RMP.2009,Lewenstein+3.PRA.2000,Lewenstein+3.PRA.2001}).
These are observables, whose average value is non-negative on all
separable states, whereas it is strictly negative on an entangled
state. Witnesses detect some entangled states only, but their
advantage is that they can be directly measured using local optimized measurements
\cite{Guhne+6.JMO.2003}. Generalizations of witnesses and
maps to the non-linear witnesses and maps are possible, but rare
and not so frequently used (cf. \cite{4Horodeckis.RMP.2009}).

Sufficient criteria of separability (i.e., necessary entnaglement
criteria) are not so common. They appeared very early in the theory of
entanglement in the studies of {\it robustness of entanglement}
\cite{Vidal.PRA.1999}. More importantly, sufficient sepaprability
criteria turn out to be crucial in the studies of volume of
separable states, which amounts to estimations of the size of the
ball around the separable states that contains separable states
only \cite{Zyczkowski+3.PRA.1998}. These ideas proved being
important  for the general aspects of quantum computing---indeed
they allowed Braunstein {\it et al.} to demonstrate that the high
temperature NMR "quantum computing" does not involve entanglement
\cite{Braunstein+5.PRL.1999}. This important problem has been
followed by many researchers---"classical" results were obtained
for instance  by  Barnum and
Gurvits \cite{Gurvits+Barnum.PRA.2002}; Later, in  the series of
works by Szarek {\it et al.}
\cite{Szarek.PRA.2005,Aubrun+Szarek.PRA.2006} powerful estimations
of the volume of separable states were made for multipartite
systems.

Here we present alternative and general approach to derive
sufficient criteria for separability and more generally of
belonging to a certain class of entanglement, quantum correlations
etc. Alternatively, we will derive conditions when a state of a
certain form (corresponding to a mixture of the completely
polarized state with pure state projectors, partially transposed
and locally transformed pure state projectors etc.) belongs to a
certain entanglement/correlation  class. Our approach is somewhat
related to that of Barnum and Gurvits: we consider families of
maps acting on global states of a bi- or multiparty system, having
the property that when  applied to a state (a state of a certain
entanglement class) they produce  a separable state (a state of
another certain entanglement class).

The paper is organized as follows. In Section \ref{Preliminaries}
we present our notation, definitions, and preliminary facts. In
Section \ref{General} we explain the general theorem, on which we
then base our sufficient criteria, and sketch the previously known
results -- in particular we remind the reader the results of
Vidal, Barnum and Gurvits. Main results for $2\otimes N$ and
$M\otimes N$ systems are presented in Sections~\ref{Sec:2xN} and
\ref{Sec:MxN}, respectively. We focus mainly on sufficient
separability criteria, but also discuss some results that go
beyond this paradigm. We close the paper with the
conclusions and outlook in Section~\ref{Conclusion}.

%=================================================== Sec-II ======================================================================%
\section{\label{Preliminaries} Preliminaries}

We consider states of a  bipartite system of
Alice and Bob, described by density operators $\rho$ acting on the
Hilbert space 
${\cal H}={\cal H}_A\otimes{\cal H}_B=
\CC^M\otimes \CC^N$ with $M\leq N$. The space of all bounded
operators acting on it is denoted by ${\cal B}({\cal H}_A\otimes{\cal H}_B)$,
while the set of states, that is, operators from ${\cal B}({\cal H}_A\otimes{\cal H}_B)$ 
obeying $\rho\ge 0$, $\rho =\rho^\dag$
and ${\rm Tr}(\rho)=1$, as $R$ (also denoted as $\Sigma_M$, see
below). Note that $R$ is convex and compact. However, the
normalization condition of the states and, consequently, the
compactness of $R$ will only be assumed if necessary. The
transposed operator $\rho$ will be denoted as $\rho^T$, whereas
partially transposed as $\rho^{T_A}$($\rho^{T_B}$). In the following we will
consider also various classes of states such as:
\begin{itemize}

\item Separable states $\Sigma$, also denoted as $\Sigma_1$, i.e.,
states that admit the following decomposition
\begin{equation}
\sigma=\sum_k^K p_k |\Psi_k\rangle_1 \langle \Psi_k|_1,
\label{sigma}
\end{equation}
where $|\Psi_k\rangle_1= |e_k\rangle \otimes |f_k\rangle$ are
simple tensors, i.e., product also know as vectors with Schmidt
rank 1, i.e., having only one non-vanishing Schmidt coefficient
\cite{Peres.S.1995}.

\item States with Schmidt number $n\le M$, $\Sigma_n$, i.e., the states that admit
the decomposition
\begin{equation}
\sigma=\sum_k^K p_k |\Psi_k\rangle_n \langle \Psi_k|_n,
\label{sigma}
\end{equation}
where $|\Psi_k\rangle_n= \sum_{l=1}^n \lambda_l|e_{kl}\rangle
\otimes |f_{kl}\rangle$ are vectors with Schmidt rank $\le n$, and
at least one of them has Schmidt rank equal to $n$
\cite{Terhal+Horodecki.PRAR.2000,Sanpera+2.PRAR.2001}.

\item PPT-operators $W_{PPT}$, i.e., those that have a positive patrial transpose, $w^{T_A}\ge
0$; in particular, we will talk about PPT-states, $R_{PPT}$, i.e.,
$\rho\in R$ and $\rho\in W_{PPT}$, i.e., $\rho^{T_A}\ge 0$.

\item Pre-witnesses of entangled states $W$, also denoted as $W_1$, i.e.,
observables $W=W^\dag$ such that for every $w\in W$ and every
$\sigma \in \Sigma$, we have ${\rm Tr}(w\sigma)\ge 0$. Note that
the genuine witnesses additionally must not be positive definite,
i.e., $w \notin R=\Sigma_M$ \cite{4Horodeckis.RMP.2009,Lewenstein+3.PRA.2000}

\item Pre-witnesses of entangled states of Schmidt number $n$, $W_n$, i.e.,
observables $W=W^\dag$ such that for every $w\in W_n$ and every
$\sigma_n \in \Sigma_n$, we have ${\rm Tr}(w_n\sigma_n)\ge 0$.
Note that the genuine witnesses additionally must detect  a state
from $\Sigma_{n+1}$, there must exist a $\sigma_{n+1}\in
\Sigma_{n+1}$, such that ${\rm Tr}(w_{n+1}\sigma_n) < 0$
\cite{Terhal+Horodecki.PRAR.2000,Sanpera+2.PRAR.2001}.

\item Decomposable entanglement pre-witnesses, i.e., pre-witnesses
that admit a decomposition $W=P+ Q^{T_A}$, where $P, Q \ge 0$.
Witnesses that do not admit such decomposition are termed
non-decomposable (cf.
\cite{Lewenstein+3.PRA.2000,Lewenstein+3.PRA.2001}).

\end{itemize}

Consequently, we will consider various classes of linear maps \[
\Lambda\colon{\cal B}({\cal H}_A\otimes{\cal H}_B) \to {\cal
B}({\cal H}_A\otimes{\cal H}_B)\] that preserve hermiticity. We
will pay particular attention to those maps that map a certain subset
of $S\subset {\cal B}({\cal H}_A\otimes{\cal H}_B)$ to another
certain subset of $S'\subset {\cal B}({\cal H}_A\otimes{\cal
H}_B)$, i.e., for every $w \in S$ we have $\Lambda(w)\in S'$. We
will typically consider families of maps $\Lambda_{\bf p}$,
parameterized  by a set of parameters ${\bf p}=(\alpha, \beta,
\gamma, \ldots)$, and assume that the map is invertible for almost
all values of the parameters.  We will denote by $R(\rho)$,
$K(\rho)$, and $r(\rho)={\rm dim} R(\rho)$ the range, the kernel
and the rank of $\rho$, respectively. Complex conjugation will be
denoted according to the  "physical tradition" as $*$, so that the
complex conjugate of a vector $|e\rangle$ is $|e^*\rangle$, while
any perpendicular vector by $|e^{\perp}\rangle$. Because of the
reasons that will become clear below we will mostly work with even
$M$ and $N$.  In such case it is straightforward to implement in an
easy way, the so called, Breuer-Hall unitary operators $V$,
$VV^{\dag}=\id$, such that for all vectors $|e\rangle$,
\begin{equation}
\label{breuer-hall}
 V|e^*\rangle=|e^{\perp}\rangle.
\end{equation}

%=================================================== Sec-III ======================================================================%
\section{\label{General} General Theorems}

Our approach is mainly based on the following simple observation.

\begin{thm}
\label{thm:general} Let  $S, S'$ are convex and compact subsets of
${\cal B}({\cal H}_A\otimes{\cal H}_B)$, and let $\Lambda_{\bf p}\colon S\to S'$
be a family of maps, invertible for almost all $\bf p$. Let
${\cal P}_{SS'}$ denote a subset of the parameters set, and the
maps have the property that for every $w\in S$ we have that
$\Lambda_{\bf p}(w)\in S'$, provided ${\bf p}\in {\cal P}_{SS'}$.
Then $\Lambda^{-1}_{\bf p}(\sigma) \in S$ $\Rightarrow$
$\sigma \in S'$.
\end{thm}
\begin{proof}
Note that $\Lambda_{\bf p}\left[\Lambda^{-1}_{\bf p}(\sigma)\right]=\sigma$. 
\end{proof}
Now let $S$ be the set of all $M\otimes N$ states and $S'\subset S$ be the separable states. This general theorem can be applied to derive sufficient separability criterion, provided the choice of $\Lambda_{\bf p}$ is such that
\begin{itemize}
	\item[i)]  we can easily
	check that $\Lambda_{\bf p}^{-1}(S)\subset S$,
	\item[ii)] we can prove the
	assumption that $\Lambda_{\bf p}(S)\subset S'$.
\end{itemize}   Clearly, the difficulty of deriving the sufficient separability criteria
is hidden in the condition ii).

Below we provide some examples (of sufficient separability criteria) known from the literature, to illustrate concrete use of Theorem \ref{thm:general}. Let us consider the simple reduction-type map $\Lambda_{\alpha}(\rho):=\tr(\rho)\id+\alpha\rho$. For condition ii), we have to find the range of the parameter $\alpha$ such that $\Lambda_{\alpha}(\rho)$ is separable. Once this is done, $\Lambda_{\alpha}^{-1}(\rho)\geq 0$ would be a sufficient separability criterion.  Separability of $\Lambda_{\alpha}(\rho)$ is well known and given by the following result.

\begin{thm}
\label{thm:barnum} \cite{Vidal.PRA.1999,Gurvits+Barnum.PRA.2002} Let $\Lambda_\alpha(\rho)=
{\rm Tr}(\rho)\id
+ \alpha \rho $ be the family of maps, and
$-1\le\alpha \le 2$. Then $\rho\ge 0$ $\Rightarrow$ $\Lambda_\alpha(\rho)=:\sigma \in \Sigma$,
 i.e. $\sigma$ is separable.
\end{thm}

\begin{proof}
It is enough to prove the theorem for pure states, $\rho=|\Psi\rangle\langle
\Psi|$. For $M=2$, $|\Psi\rangle$ has maximally Schmidt rank 2,
and, without loosing generality, we can assume that $|\Psi\rangle
=\lambda_0 |0\rangle\otimes |0\rangle + \lambda_1|1\rangle\otimes
|1\rangle$. It is then enough to check positivity and separability
of $\Lambda(|\Psi\rangle\langle \Psi|)$ on a $2\otimes 2$ space
spanned by $|0\rangle$, $|1\rangle$ in both Alice's and Bob's
spaces. PPT provides then necessary and sufficient condition
\cite{3Horodecki.PLA.1996}, and we easily get that indeed $-1\le
\alpha\le 2$. The condition $\alpha \ge -1$ actually follows
already from a simpler requirement of positivity of
$\Lambda_{\alpha}(\rho)$, as it appears in the definition of the reduction
map \cite{M+PHorodecki.PRA.1999,Cerf+2.PRA.1999}. 
%
%The proof for $M\ge 3$ is more
%complex, the details can be found in Refs.~\cite{Vidal.PRA.1999,Gurvits+Barnum.PRA.2002}.
%We will give an alternative proof in Section~\ref{Subsection:M.N.Reduction-Breuer}, while dealing with arbitrary 
%$M,N$. Our original proof is based on
%decomposition of $\Lambda_{\alpha}(\rho)$  into an explicitly separable
%state and a diagonal positive definite, i.e. also separable matrix.
%

To prove the theorem for $M\geq 3$ we exploit the fact, proven in Ref. \cite{Vidal.PRA.1999} (see also Ref. \cite{Czekaj}) 
that the (unnormalized) mixture of $\ket{\psi}$ and the maximally mixed state, 
%
%\begin{equation}
$\mathbbm{1}+\alpha\ket{\psi}\!\bra{\psi}$,
%\end{equation}
is separable iff 
\begin{equation}
\alpha\leq \frac{1}{\lambda_0\lambda_1}.
\end{equation}
This implies that $\mathbbm{1}+\alpha\ket{\psi}\!\bra{\psi}$ is separable 
for any $\ket{\psi}$ if $\alpha\leq 2$.
\end{proof}
This leads to the following separability criteria.

\begin{thm}
\label{cor:barnum1} {\bf (Sufficient separability criterion 1)} If
$\Lambda^{-1}_\alpha(\sigma):= [\sigma -{\rm
Tr}(\sigma)\id/(2N+\alpha)]/\alpha \ge 0$ then $\sigma\in \Sigma$,
i.e. $\sigma$ is separable.
\end{thm}
 In particular, for the extreme values of $\alpha$, we get the following separability results. 
\begin{cor}
\label{cor:barnum2} If $\sigma -{\rm
Tr}(\sigma)\id/(2N+2) \ge 0$ then $\sigma\in \Sigma$, i.e.
$\sigma$ is separable.
\end{cor}
\begin{cor}
\label{cor:barnum3} If ${\rm
Tr}(\sigma)\id/(2N-1)-\sigma \ge 0$ then $\sigma\in \Sigma$, i.e.
$\sigma$ is separable.
\end{cor}

It is amazing that these  simple criteria are strong enough to detect some states outside the "separable ball" around identity \cite{Gurvits+Barnum.PRA.2002} (see Section~\ref{Subsection:examples.2.N} for an explicit example and detailed discussion). Of course, not all states from the separable ball are detected. The reason for this weakness is that Barnum and Gurvits  have estimated the radius of the separable ball using a stronger result, namely: if the operator norm of a Hermitian matrix $X$ is bounded by 1, then $\id+ X$ is separable. In contrast, our simple criteria depend on the spectrum, not on the norm, and are thereby independent. Nonetheless, we will see in the next sections that with other choices of $\Lambda_{\bf p}$'s, we could derive further stronger sufficient separability criteria.

%=================================================== Sec-IV ======================================================================%
\section{\label{Sec:2xN}Main Results for $\CC^2\otimes\CC^N$}

\subsection{Reduction- and Breuer-Hall-like maps}

Let us start by presenting generalizations of Theorem~\ref{thm:barnum} 
derived in a similar spirit as the
generalization of the reduction map by the Breuer-Hall map
\cite{Breuer.PRL.2006,Hall.JPA.2006}. Note that Breuer-Hall's construction does not work
in the qubit case---Breuer-Hall's map identically vanishes for qubits. In our
case when we consider maps acting on the composite space of Alice
and Bob this restriction will not apply. In the two-dimensional
space there exist (up to a phase) a single unitary,
$\sigma_2$, with the property that for every $|e\rangle$,
$\sigma_2|e^*\rangle=|e^{\perp}\rangle$. For simplicity we will
consider the case when Bob's space has even dimension; in such a
case it is easy to construct analogous Breuer-Hall unitary
operator $V$, such that $V|f^*\rangle=|f^{\perp}\rangle$.

Below we will denote, for brevity: \[\tilde\rho_A=\sigma_2\rho^{T_A}\sigma_2,\]
\[\tilde\rho_B=V\rho^{T_B}V^\dag,\] and
\[\tilde\rho=\sigma_2V\rho^{T}V^\dag\sigma_2.\]
We can now prove the following result about a generalized $\Lambda_p$.

\begin{thm}
\label{thm:lewenstein1} Let $\Lambda_{\alpha,\beta}(\rho)= {\rm
Tr}(\rho)\id + \alpha \rho + \beta \tilde\rho_A$ be a family of
maps. Let $\alpha \ge {\rm max}[-1,\beta/2-1]$ and $\beta \ge
{\rm max}[-1,\alpha/2-1]$. Then $\rho\ge 0$ $\Rightarrow$
$\Lambda_{\alpha,\beta}(\rho)=:\sigma\in \Sigma$, i.e., $\sigma$ is
separable.
\end{thm}
\begin{proof}
The proof is very similar to that of Theorem \ref{thm:barnum} as above. 
It is enough to prove it for pure
states, $\rho=|\Psi\rangle\langle \Psi|$. Without any loss of generality, 
take $|\Psi\rangle =\lambda_0 |0\rangle\otimes
|0\rangle + \lambda_1|1\rangle\otimes |1\rangle$. It is enough to
check then positivity and separability of
$\Lambda(|\Psi\rangle\langle \Psi|)$ on a $2\otimes 2$ space
spanned by $|0\rangle$, $|1\rangle$ in both Alice's and Bob's
spaces, where PPT provides then necessary and sufficient
condition. After some algebra we obtain the desired condition,
represented graphically in Fig.~\ref{Fig:2xn.alpha.beta.range},
where we mark the range of $\alpha,\beta$ for which
$\Lambda_{\alpha,\beta}(\rho)$ is separable for any $\rho$.
\begin{figure}[h]
    \begin{center}
        \includegraphics[width=6cm,height=6cm]{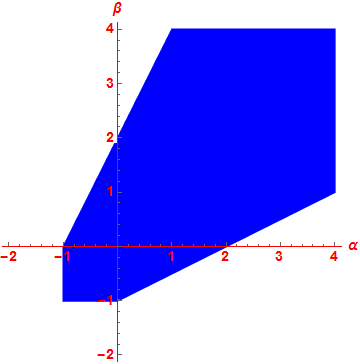}\caption{The range of $\alpha,\beta$ such that for any 
        state $\rho$ acting on $\CC^2\otimes \CC^N$, $\Lambda_{\alpha,\beta}(\rho)$ is separable.}
        \label{Fig:2xn.alpha.beta.range}
    \end{center}
\end{figure}
\end{proof}
The general Theorem~\ref{thm:general} then immediately implies the following separability result.
\begin{thm}
\label{thm:lewenstein2} {\bf (Sufficient separability criterion 2)}
Let $\rho =\Lambda_{\alpha,\beta}^{-1}(\sigma)\ge 0 $, i.e.,
\begin{equation}\label{MainResult:1}
\frac{1}{\alpha^2-\beta^2}\left[\alpha \sigma -\beta
\tilde\sigma_A - \frac{(\alpha-\beta)}{2N+\alpha +\beta} {\rm
Tr}(\sigma) \id\right]\ge 0.
\end{equation}
Then $\sigma$ is separable.
\end{thm}

Noticeably, the contributions from $\alpha$ and $\beta$ compensate
each other, diminishing in effect off-diagonal parts of $\sigma$
in the Alice space.  In fact, for $\alpha=\beta$ these off-diagonal
parts strictly vanish and $\sigma$ is a sum of two simple tensor
product matrices---in this special case, however, the map ceases
to be invertible. The criterion derived above is
particularly strong at the boundary of the parameter region
represented in Fig.~\ref{Fig:2xn.alpha.beta.range}, i.e., for
$\beta=\alpha/2-1$ or $\alpha=\beta/2-1$.
\begin{cor}
\label{cor:lewenstein2} {\bf (Sufficient separability criterion 3)}
Let
\begin{equation}\label{MainResult:2a}
\left[\alpha\sigma -(\alpha/2-1) \tilde\sigma_A -
\frac{(1+\alpha/2)}{2N+ 3\alpha/2-1} {\rm Tr}(\sigma)
\id\right]\ge 0
\end{equation}
or
\begin{equation}\label{MainResult:2b}
\left[\beta\tilde\sigma_A -(\beta/2-1) \sigma -
\frac{(1+\beta/2)}{2N+ 3\beta/2-1} {\rm Tr}(\sigma) \id\right]\ge
0.
\end{equation}
Then $\sigma$ is separable.
\end{cor}

The above corollary has a particularly interesting limit, when both
$\alpha$ and $\beta$ tend to $\infty$.
\begin{cor}
\label{cor:lewenstein3} {\bf (Sufficient separability criterion 4)}
Let
\begin{equation}\label{MainResult:3a}
\sigma - \tilde\sigma_A/2 > 0
\end{equation}
or
\begin{equation}\label{MainResult:3b}
\tilde\sigma_A - \sigma/2 > 0.
\end{equation}
Then $\sigma$ is separable.
\end{cor}

Note that the latter inequalities are valid only in the asymptotic
sense and require strict positivity. One should also note that in
the above results we used maps that
 explicitly involved partially transposed matrices. For this
reason it is useful to remind the reader some results of Ref.
\cite{Kraus+3.PRA.2000} that hold for $M=2$.

\begin{thm}
\label{thm:Karnas} \cite{Kraus+3.PRA.2000} If $\sigma=\sigma^{T_A}\ge 0$
then
 $\sigma$ is separable.
\end{thm}

\begin{cor}
\label{cor:Karnas} \cite{Kraus+3.PRA.2000} Let $\sigma\ge 0$ be a state.
If $\sigma + \sigma^{T_A}$ is of full rank and $\|(\sigma +
\sigma^{T_A})^{-1}\|\ \| \sigma - \sigma^{T_A}\| \le 1$, then
$\sigma$ is separable. This corollary implies that if $\sigma$ is
of full rank and is very close to $\sigma^{T_A}$ then
 $\sigma$ is separable. Here we use the operator norm
$\|A\| := \max_{\|\Psi\rangle\|=1} \| A |\Psi\rangle\|$.
\end{cor}

Note that the present results  are  clearly independent since they
involve matrices $\tilde\sigma_A$ and related ones. Nonetheless, explicit
examples showing the independence will be given in Section~\ref{Subsection:examples.2.N}.

Let us now present the strongest theorem of this section, which
involves Breuer-Hall unitary operators on the both sides of Alice
and Bob:

\begin{thm}
\label{thm:lewenstein2} Let ${\bf p} =\{\alpha, \beta, \gamma,
\delta\}$ and $\Lambda_{\bf p}(\rho)= {\rm Tr}(\rho)\id + \alpha
\rho + \beta \tilde\rho_A + \gamma \tilde\rho_B  + \delta
\tilde\rho$ be the family of maps; let there exist $0\le a, b\le
1$, $a+b\le 1$ such that the parameters fulfill the four conditions:
\begin{subequations}\label{Eq:cond.alpha-delta.a-b}
	\begin{align}
	\alpha &\ge \beta/2-a, \\
	\beta  &\ge \alpha/2-a, \\
	\gamma &\ge \delta/2-b, \\
    \delta &\ge \gamma/2-b,
	\end{align}
\end{subequations}
and $\alpha\ge -1$, $\beta\ge -1$, $\gamma\ge -1$, and $\delta\ge
-1$. Then $\rho\ge 0$ $\Rightarrow$ $\sigma = \Lambda_{\bf
p}(\rho)\in \Sigma$, i.e. $\sigma$ is separable.
\end{thm}

\begin{proof}
The proof is similar to the proof of Theorem~\ref{thm:lewenstein1},  
but more complex and technical. Again it is enough to prove for
pure states, $\rho=|\Psi\rangle\langle \Psi|$ and, without loosing
generality, take $|\Psi\rangle =\lambda_0 |0\rangle\otimes
|0\rangle + \lambda_1|1\rangle\otimes |1\rangle$. But, now we have
to consider three cases: i) when the Breuer-Hall unitary $V$ acts
as $\sigma_2$ on the Bob's subspace spanned by $|0\rangle$ and
$|1\rangle$; ii) when $V$ transforms $|0\rangle$ and $|1\rangle$
to $|2\rangle$ and $|3\rangle$; iii) when $V$ transforms
$|0\rangle$ and $|1\rangle$ to orthogonal vectors in a subspace
spanned by $|0\rangle$ and  $|2\rangle$, and $|1\rangle$ and
$|3\rangle$, respectively.

The conditions $\alpha\ge -1$, $\beta\ge -1$, $\gamma\ge -1$, and
$\delta\ge -1$ follow again from the analysis of the case when
$\rho$ is a projector on a product state, say
$|0\rangle\otimes|0\rangle$, so that $\tilde\rho_A$ is a projector
on $|1\rangle\otimes|0\rangle$, $\tilde\rho_B$ on
$|0\rangle\otimes|\tilde 1\rangle$, and $\tilde\rho$ on
$|1\rangle\otimes|\tilde 1\rangle$, where $|\tilde 1\rangle$ is a
vector orthogonal to $|0\rangle$. To derive the conditions
in Eq.~\eqref{Eq:cond.alpha-delta.a-b}, we split $\id= a\id + b\id + (1-a-b)\id$ 
and apply the result of Theorem \ref{thm:lewenstein1} to $a\id + \alpha\rho
+\beta\tilde \rho_A$, and $(1-a)\id + \gamma\tilde\rho_B
+\delta\tilde\rho$. Clearly, this estimate of the region of
parameters, where $\Lambda_{\bf p}(\rho)$ is separable, is very
conservative, and probably can be improved significantly.
\end{proof}

The conditions above  are sufficient for separability and
correspond to the case iii), which probably is not the most
demanding, but we were not able to find weaker sufficient
conditions, i.e., the largest allowed regions of the parameters.
The case i) leads to the less demanding restrictions
that follow from the above conditions, when we set $a+b=1$.
Obviously, this case is of interests {\it per se}, since it
corresponds to exact conditions in the case of two qubits, so that
we present it as a separate corollary below. Finally, the case
ii) is the least restrictive -- indeed it allows to exceed the
restriction $a+b\le 2$ and reach $a=b=1$.

\begin{cor}
\label{cor:lewenstein4} Let  $\Lambda_{\bf p}(\rho)= {\rm
Tr}(\rho)\id + \alpha \rho + \beta \tilde\rho_A + \gamma
\tilde\rho_B  + \delta \tilde\rho$ be the family of maps acting on
${\cal H}={\cal H}_A\otimes{\cal H}_B=\CC^2\otimes \CC^2$, with $V=\sigma_2$ (the second Pauli matrix) acting in the
Bob's space. Let $s=\alpha+\gamma$, $\tilde s=\beta + \delta$, and
let the parameters fulfill the conditions:
\begin{subequations}\label{Eq:2x2cond1}
	\begin{align}
	s &\ge \max[-1,\tilde s/2-1], \\
	\tilde{s} &\ge {\rm max}[-1,s/2-1],
	\end{align}
\end{subequations}
and $\alpha\ge -1$, $\beta\ge -1$, $\gamma\ge -1$, and $\delta\ge -1$.
Then if $\rho\ge 0$, $\sigma = \Lambda_{\bf
p}(\rho)\in \Sigma$, i.e. $\sigma$ is separable.
\end{cor}
\begin{proof}
The proof is straightforward, since,  PPT provides the if and only
if criterion for separability in the case 2 by 2. We take
$\rho=|\Psi\rangle\langle \Psi|$ with $|\Psi\rangle =\lambda_0
|0\rangle\otimes |0\rangle + \lambda_1|1\rangle\otimes |1\rangle$,
and simply write down the matrices $\Lambda_{\bf p}(\rho)$ and
$\Lambda_{\bf p}(\rho)^{T_A}$ and check positive definiteness. We
observe that it is enough to consider the "extreme" cases, when
either $\lambda_0=0$ or $\lambda_1=0$, or when $|\Psi\rangle$ is
maximally entangled, i.e. $\lambda_0=\lambda_1=1/\sqrt{2}$. The
latter case gives us conditions for $s$ and $\tilde s$, whereas
the former the simple lower bounds by -1 on all parameters.
To be more explicit the conditions $\alpha\ge -1$, $\beta\ge -1$, $\gamma\ge -1$, and
$\delta\ge -1$ follow from the analysis of the case when $\rho$ is
a projector on a product state, say $|0\rangle\otimes|0\rangle$,
so that $\tilde\rho_A$ is a projector on
$|1\rangle\otimes|0\rangle$, $\tilde\rho_B$ on
$|0\rangle\otimes|1\rangle$, and $\tilde\rho$ on
$|1\rangle\otimes|1\rangle$. Let us explicitly write $\Lambda_{\bf
p}(\rho)= {\rm Tr}(\rho)\id + \alpha \rho + \beta \tilde\rho_A +
\gamma \tilde\rho_B  + \delta \tilde\rho$, for
$\rho=|\Psi\rangle\langle \Psi|$ with $|\Psi\rangle =\lambda_0
|0\rangle\otimes |0\rangle + \lambda_1|1\rangle\otimes |1\rangle$:
\begin{widetext}\[
 \Lambda_{\bf
p}(\rho)= \begin{pmatrix}
1+\alpha\lambda_0^2+\delta\lambda_1^2 & 0 & 0 & (\alpha-\beta+\delta-\gamma)\lambda_0\lambda_1 \\
0 & 1+\gamma\lambda_0^2+\beta\lambda_1^2 & 0 & 0 \\
0 & 0 & 1+\beta\lambda_0^2+\gamma\lambda_1^2 & 0  \\
(\alpha-\beta+\delta-\gamma)\lambda_0\lambda_1 & 0 & 0 &
1+\delta\lambda_0^2+\alpha\lambda_1^2
\end{pmatrix}.
\]
\end{widetext} It is easy to observe that the positivity of
$\Lambda_{\bf p}(\rho)$ and its partial transpose (i.e. its
separability),  gives the strongest restrictions on the parameters
when $|\Psi\rangle$ is maximally entangled, i.e. $\lambda_0=\lambda_1=1/\sqrt{2}$. Direct
inspection then leads to the conditions in Eq.~\eqref{Eq:2x2cond1}.
\end{proof}
Note that the proof of the Theorem \ref{thm:lewenstein1} is just
the same,  when we set $\gamma=\delta=0$, while  the proof of 
the Theorem~\ref{thm:lewenstein2}, as discussed above, is more involved.

\subsection{\label{Subsection:Andolike.2.N}Ando-like maps}

So far we were mainly generalizing the Breuer-Hall type maps in the context of our general approach. Let us now present a similar separability criterion based on another family of positive maps, the so called Ando's maps \cite{Tanahasi.CMB.1988,Osaka.LAA.1991-93,Ha.PRIMS.1998}.
For $k=1,2,\dotsc,d-1$, the map $\Lambda_{k;\alpha}\colon{\cal B}({\CC}^d)\to {\cal B}({\CC}^d)$ is defined by
\[\Lambda_{k;\alpha}(\rho):=(d-k)\,\epsilon(\rho)+\sum_{l=1}^k\epsilon\left(S^l\rho S^{l\dagger}\right)+\alpha\rho,
\]
where $S|i\rangle=|i+1\rangle$ is a unitary shift modulo $d$, and $\epsilon(X):=\sum_{i=0}^{d-1}|i\rangle\langle i|X|i\rangle\langle i|$. It is known that $\Lambda_{k;\alpha}$ is positive for $\alpha\geq -1$ \cite{Tanahasi.CMB.1988,Osaka.LAA.1991-93,Ha.PRIMS.1998}. Also note that $\Lambda_{d-1;\alpha}=\id\tr(\rho)+\alpha\rho$ is the generalized reduction map already considered in Theorem~\ref{thm:barnum}.

The Ando-type maps defined above do not consider the internal tensor-product structure of the space ${\cal B}({\CC}^d)$. To make it transparent, we define for $\rho\in\mathcal{B}\left(\CC^2\otimes\CC^N\right)$  and $-1\le k\le N-1$,\begin{eqnarray}
\Lambda_{k,\alpha}^{2\times N}(\rho)&:=&(N-k-1)\,\epsilon(\rho)+\sum_{m=0}^{N-1}\epsilon\left(S_B^m\rho S_B^{m\dagger}\right)\nonumber\\
&+&\sum_{m=0}^{k}\epsilon\left(S_A S_B^m\rho S_A^{\dagger}S_B^{m\dagger}\right)+\alpha\rho,
\end{eqnarray}
where
 \[\epsilon(X):=\sum_{i=0, j=0}^{1,N-1}|i,j\rangle\langle i,j|X|i,j\rangle\langle i,j|,\]
 and $S_A$, $S_B$ are the unitary shift operator on the Alice and Bob spaces, respectively.  
The case $k=-1$ is particularly simple, since the unitary shift is applied on Bob's side only. For simplicity, we denote it differently,
$\Lambda_{-1}^{2\times N}(.)\equiv\Lambda^{2\times N}(.)$, and discuss below its separability properties.  Its generalization for $M\ge 3$ is presented in Section~\ref{SubSec:Ando.arbitary.M.N}.

\begin{thm}
	\label{thm:Darek.Andos.map1} For $\rho\in\mathcal{B}\left(\CC^2\otimes\CC^N\right)$, \[\sigma:=\Lambda^{2\times N}(\rho)=N\,\epsilon(\rho)+\sum_{m=0}^{N-1}\epsilon\left(S_B^m\rho S_B^{m\dagger}\right)\nonumber+\alpha\rho\]
	is separable if $-2N/(3N-1)\le \alpha\le 1$.
\end{thm}
\begin{proof} 
As usual, we can  restrict ourselves to consider $\rho=|\psi\rangle\langle\psi|$, and without loosing generality assume that   $|\psi\rangle=
 |0\rangle|e\rangle + |1\rangle|f\rangle$, with $|e\rangle=\sum_{j=0}^{N-1}\lambda_{0j}|j\rangle$, 
 $|f\rangle=\sum_{j=0}^{N-1}\lambda_{1j}|j\rangle$, and $\sum_{j=0}^{N-1}(|\lambda_{0j}|^2 + |\lambda_{1j}|^2)=1$.
   One gets (we present the results for $N=3$,  but generalization for $N>3$ is straightforward)
\begin{widetext}
\begin{align*}
\Lambda^{2\times 3}(\rho) &=\left[ \begin{array}{ccc|ccc} 3|\lambda_{00}|^2 +\sum\limits_{j=0}^{2}|\lambda_{0j}|^2  & 0 & 0 &  \alpha\lambda_{00}\lambda_{10}^* & ... & \alpha\lambda_{00}\lambda_{12}^* \\ 0 &  3|\lambda_{01}|^2 +\sum\limits_{j=0}^{2}|\lambda_{0j}|^2  & 0 & ... & ...  & ...\\ 
0 &  0 & 3|\lambda_{02}|^2 +\sum\limits_{j=0}^{2}|\lambda_{0j}|^2 & \alpha\lambda_{02}\lambda_{10}^*& ... & \alpha\lambda_{02}\lambda_{12}^* \\ 
\hline \alpha\lambda_{00}^*\lambda_{10}& ... & \alpha\lambda_{02}^*\lambda_{10}& 3|\lambda_{10}|^2 + \sum\limits_{j=0}^{2}|\lambda_{1j}|^2& 0 & 0 \\ ... &... & ...&  0 &  3|\lambda_{11}|^2 + \sum\limits_{j=0}^{2}|\lambda_{1j}|^2 & 0 \\
\alpha\lambda_{00}^*\lambda_{12}& ... & \alpha\lambda_{02}^*\lambda_{12}& 0 & 0 & 3|\lambda_{12}|^2 + \sum\limits_{j=0}^{2}|\lambda_{1j}|^2 
\end{array} \right]\\
+&  \alpha|0\rangle\langle 0|\otimes |e\rangle\langle e| + \alpha |1\rangle\langle 1|\otimes |f\rangle\langle f|.
\end{align*}
\end{widetext}
The result is separable iff it is non-negative and PPT, but in the proof we will not use the PPT for the sake of generalization. Instead, we subtract from  $\Lambda^{2\times 3}(\rho)$ a set of explicitly separable (unnormalized) states $\sigma_{i,j}$ (with $i<j=0,1,2$), which are supported in $2\otimes 2$ (thus nonnegativity and PPT implying separability) and are of the form 
 \begin{widetext}
 \begin{equation}\label{sigmaij}
\sigma_{i,j }= \left[ \begin{array}{cc|cc}|\alpha|( |\lambda_{0i}|^2 +|\lambda_{0j}|^2)  & 0 &0 &\alpha \lambda_{0i}\lambda_{1j}^*\\
 0 & |\alpha|(|\lambda_{0i}|^2 +|\lambda_{0j}|^2 ) &\alpha \lambda_{0j}\lambda_{1i}^* &0\\
\hline
 0 & \alpha \lambda_{0j}^*\lambda_{1i} & |\alpha|(|\lambda_{1i}|^2 +|\lambda_{1j}|^2 ) & 0\\
 \alpha \lambda_{0i}^*\lambda_{1j} & 0 & 0 & |\alpha|(|\lambda_{1i}|^2 +|\lambda_{1j}|^2 )\end{array} \right] .
\end{equation}
\end{widetext}

The remainder  $R(\rho):=\Lambda^{2\times 3}(\rho) -\sum_{i<j=0}^2\sigma_{i,j}$ is given by 
\begin{equation}
R(\rho)=\left[\begin{array}{c|c}
D_0&X\\
\hline
X^\dagger&D_1
\end{array}\right]+\alpha |0\rangle\langle 0|\otimes |e\rangle\langle e| + \alpha |1\rangle\langle 1|\otimes |f\rangle\langle f|,
\end{equation}
\begin{align*}
\text{where~ } D_k&=\diag\left\{(3-|\alpha|)|\lambda_{ki}|^2 +(1-|\alpha|)N_k\right\}_{i=0}^2,\,k=0,1,\\
X&=\diag\left\{\alpha\lambda_{00}\lambda_{10}^*,\,\alpha\lambda_{01}\lambda_{11}^*,\,\alpha\lambda_{02}\lambda_{12}^*\right\},\\
N_0&=1-N_1=\sum_{j=0}^2|\lambda_{0j}|^2.
\end{align*}
For $0\le \alpha\le 1$, $R(\rho)$ is explicitly non-negative and separable. For $\alpha< 0$, we again subtract separable states
of the form $|\alpha|$ $|e_i\rangle\langle e_i|\otimes|i\rangle\langle i|$ ($i=0,1,2$), with  $|e_i\rangle=\lambda_{0i}|0\rangle-\lambda_{1i}|1\rangle$, and end up with 
\begin{equation}\label{Eq:def.R'}
R'(\rho)=\left[\begin{array}{c|c}
D_0'&\bm{0}\\
\hline
\bm{0}&D_1'
\end{array}\right]+\alpha |0\rangle\langle 0|\otimes |e\rangle\langle e| + \alpha |1\rangle\langle 1|\otimes |f\rangle\langle f|,
\end{equation}
\[\text{where~ } D_k'=\diag\left\{(3-2|\alpha|)|\lambda_{ki}|^2 +(1-|\alpha|)N_k\right\}_{i=0}^2,\,k=0,1.\]
Noticing that the RHS of Eq.~\eqref{Eq:def.R'} could be written as $|0\ran\lan 0|\otimes (D_0'+\alpha|e\ran\lan e|)+|1\ran\lan 1|\otimes (D_1'+\alpha|f\ran\lan f|)$,
$R'(\rho)$ is explicitly separable if $D_0'+\alpha|e\ran\lan e|\geq 0$ and $D_1'+\alpha|f\ran\lan f|\geq 0$.
 
After straightforward calculations using the results of Ref.~\cite{Lewenstein+Sanpera.PRL.1998}, we get that this requires
that for all choices of $\lambda$'s  we must have 
\[ \sum_{j=0}^{N-1}\frac{|\alpha||\lambda_{0j}|^2}{ (3-2|\alpha|)|\lambda_{0j}|^2   +(1-|\alpha|)N_0}\le 1.\]
This leads to  $-3/4 \le \alpha$, or generally $-2N/(3N-1)\le \alpha$. Numerical analysis and examples suggest that the regions of $\alpha$ assuring separability can be extended to $-1\le \alpha$, but we were not able to prove it.
\end{proof}
Noticing that the map $\Lambda^{2\times N}$ allows an inverse, 
the above theorem implies the following sufficient separability criterion. 
\begin{cor}
\label{cor:lewenstein4} {\bf (Sufficient separability criterion 5)}
Let $-2N/(3N-1)\le\alpha\le 1$ and 
\begin{equation*}
\frac{1}{\alpha}\left(\sigma -\frac{N}{N+\alpha}\epsilon(\sigma)-\frac{\alpha}{(N+\alpha)(2N+\alpha)}\sum_{m=0}^{N-1}\epsilon\left(S_B^m\sigma S_B^{m\dagger}\right)\right)\ge 0.
\end{equation*}
Then $\sigma$ is separable.
\end{cor}

The methods used to prove Theorem~\ref{thm:Darek.Andos.map1} can readily be applied to prove the cases of other $k$'s. The next result is such a generalization.  
\begin{thm}
	\label{thm:Darek.Andos.map} For $\rho\in\mathcal{B}\left(\CC^2\otimes\CC^N\right)$, $\Lambda_{k}^{2\times N}(\rho)$ is separable if \begin{itemize}
		\item $k=0$ and $-(2N-1)/(3N-2)\leq\alpha\leq 1$
                     \item $k=1$ and $-(2N-2)/(3N-2)\leq\alpha\leq 1$ 
                     \item for arbitrary  $k$   and $-(2N-k-1)/(3N-2)\leq\alpha\leq 1$
		\item $k=N-1$ and $-1\leq\alpha\leq 2$ (reduction map)
	\end{itemize}
\end{thm}
\begin{proof}
The case $k=N-1$ follows from the results above. We provide first the proof for $k=0$ and $N=3$ -- generalization 
to $k=0$ and $N>3$ is straightforward. We follow exactly the steps from the proof of 
Theorem~\ref{thm:Darek.Andos.map1} and first subtract the matrices $\sigma_{i,j}$ defined in Eq. (\ref{sigmaij})
from $\Lambda_0^{2\times 3}(\rho)$ to obtain
\begin{equation}\label{Rbis}
R''(\rho)=\left[\begin{array}{c|c}
 	D_0''&X\\
 	\hline
 	X^\dagger&D_1''
 \end{array}\right]
+\alpha|0\rangle\langle 0|\otimes |e\rangle\langle e| + \alpha |1\rangle\langle 1|\otimes |f\rangle\langle f|,
\end{equation}
where
\begin{eqnarray}
D''_k&=&\diag\{(2-|\alpha|)|\lambda_{00}|^2+(1-|\alpha|)N_0+|\lambda_{10}|^2,\nonumber\\
&&\qquad(2-|\alpha|)|\lambda_{01}|^2+(1-|\alpha|)N_0+|\lambda_{11}|^2,\nonumber\\
&&\qquad(2-|\alpha|)|\lambda_{02}|^2+(1-|\alpha|)N_0+|\lambda_{12}|^2\}
\end{eqnarray}
and $X$ is defined as in the proof of Theorem \ref{thm:Darek.Andos.map1}. Now, for $\alpha\geq 0$
it is not difficult to see that the first matrix in (\ref{Rbis}) is separable. Thus, $R''(\rho)$ is separable and so is 
$\Lambda_0^{2\times 3}(\rho)$.

%
%\[ \Lambda_0^{2\times 3}(\rho)= \left[\begin{array}{c|c}
% 	D_0&X\\
% 	\hline
% 	X^\dagger&D_1
% \end{array}\right]
%+\alpha|0\rangle\langle 0|\otimes |e\rangle\langle e| + \alpha |1\rangle\langle 1|\otimes |f\rangle\langle f|.\]
%
%
%We perform then exactly the same first  step as above, i.e., subtracting $2\otimes 2$ separable states to get $R(\rho)$. 
%
For $\alpha<0$, we have to modify the second step a little bit,  taking  $|e_i\rangle=\lambda_{0i}|1\rangle-\lambda_{1i}|0\rangle$. In this way we end up with the condition 
that  $\Lambda_0^{2\times 3}(\rho)$ is explicitly separable if 
\begin{widetext}
\begin{equation*}
 \left[\begin{array}{ccc} (2-|\alpha|)|\lambda_{00}|^2 +(1-|\alpha|)(N_0+|\lambda_{10}|^2)   & 0 & 0  \\ 
0 &  (2-|\alpha|)  |\lambda_{01}|^2    +(1-|\alpha|)(N_0+|\lambda_{11}|^2) & 0  \\ 
0 &  0 &  (2-|\alpha|)  |\lambda_{02}|^2   +(1-|\alpha|)(N_0+|\lambda_{12}|^2)  \end{array} \right]+\alpha |e\rangle\langle e|  \ge 0,
\end{equation*}
\end{widetext}
and the similar condition holds for the  $|f\rangle\langle f|$ part. 
The calculations using the results of Ref. \cite{Lewenstein+Sanpera.PRL.1998} are similar; we get that $|\alpha|$ must fulfill for all choices of $\lambda$'s
\[ \sum_{j=0}^{N-1}\frac{|\alpha||\lambda_{0j}|^2}{ (2-|\alpha|)|\lambda_{0j}|^2  +(1-|\alpha|)(N_0+|\lambda_{1j}|^2)}\le 1,\]
and the analogous conditions for  the  $|f\rangle\langle f|$ part. It is easy to see that the  bound on  $|\alpha|$ 
is  $-5/7\le \alpha$, or in general $-(2N-1)/(3N-2)\le \alpha$.  This kind of proof works indeed for arbitrary $k$ and leads in general to the condition  $-(2N-1-k)/(3N-2)\le \alpha$. Clearly, as $k$ approaches $N-1$ we expect that  a better proof and better estimates should be possible, but we have not found them. 
\end{proof}

\subsection{\label{Subsection:examples.2.N}Examples}
In this subsection we will present examples of separable states detected by our criteria. For the sake of comparison, we will consider other separability criteria from the literature. As such, the basic standard is to compare with the results from Ref.~\cite{Gurvits+Barnum.PRA.2002}. The strongest separability criterion derivable from Theorem~1 of Ref.~\cite{Gurvits+Barnum.PRA.2002} is the Corollary~2 (\emph{scaling} criterion) therein, which for a normalized $M\otimes N$ state $\rho$ becomes the \emph{purity condition}, \begin{equation}\label{Barnumball}
\tr(\rho^2)\leq 1/(MN-1)\Rightarrow \rho\in\Sigma.
\end{equation}
States satisfying Eq.~\eqref{Barnumball} are usually termed as members of the \emph{separable ball around identity}.
Since our criteria Eqs.~\eqref{MainResult:1}-\eqref{MainResult:3b} are mainly for $2\otimes N$ states, we use the relation $\tilde{\sigma}_A:=\sigma_2\sigma^{T_A}\sigma_2=\id\otimes\sigma_B-\sigma$ to simplify the equations a little bit. Then Eqs.~\eqref{MainResult:2a}-\eqref{MainResult:2b} become 
	\begin{subequations}
		\begin{align}
		\left(3\frac{\alpha}{2}-1\right)\sigma&\geq \left(\frac{\alpha}{2}-1\right)\id\otimes\sigma_B+\frac{1+\alpha/2}{2N+3\alpha/2-1}\id\label{Eq:Ex:2xN.geq.1},\\
		\left(3\frac{\beta}{2}-1\right)\sigma&\leq \beta\id\otimes\sigma_B-\frac{1+\beta/2}{2N+3\beta/2-1}\id\label{Eq:Ex:2xN.leq.1},
		\end{align}
	\end{subequations}
where $0\leq\alpha,\beta<\infty$. If we set $\alpha=2$ (so that $\beta=0$), Eq.~\eqref{Eq:Ex:2xN.geq.1} becomes the separability criterion mentioned in Corrolary~\ref{cor:barnum2}.

The first bound entangled state in the literature (and also in the least possible dimension, $2\otimes 4$) is given by \cite{Horodecki.PLA.1997}
\[\label{phbes24} \rho_a=\frac{7a}{7a+1}\rho_{\mathrm{ent}}+\frac{1}{7a+1}|\phi\ran\lan\phi|,\quad a\in[0,1],\] 
\begin{align*}
\phi\ran&= |1\ran\otimes\left( \sqrt{\frac{1+a}{2}}|0\ran+\sqrt{\frac{1-a}{2}}|2\ran\right),\\
\rho_{\mathrm{ent}}&=\frac{2}{7}\sum_{i=1}^3|\psi_i\ran\lan\psi_i|+\frac{1}{7}|03\ran\lan03|,\\
|\psi_i\ran&=\frac{1}{\sqrt{2}}\left(|0\ran |i-1\ran+|1\ran |i\ran\right),\quad i=1,2,3.
\end{align*}
The state $\rho_a$ remains PPT throughout $a\in[0,1]$ and for $0<a<1$, it is bound entangled.

Let us now consider the following class of states,
\begin{equation}
\label{Example:HM2x4.+.id}\rho_{a,p}=p\rho_a+\frac{1-p}{8}\id_8.
\end{equation}
One verifies (e.g., with a computer program) that for all $a,p$ in the range
\[0<a<\frac{1}{224} \left(\sqrt{3745}-49\right),\, \frac{1}{7} \sqrt{\frac{343 a^2+98 a+7}{47 a^2-6 a+7}}<p\leq \frac{1}{5},\]
$\rho_{a,p}$violates Eq.~\eqref{Barnumball} but satisfies Eq.~\eqref{Eq:Ex:2xN.geq.1} with $\alpha=2$. Hence, $\rho_{a,p}$ lies outside the separable ball, but still detected to be separable by 
the criterion mentioned in Corrolary~\ref{cor:barnum2}. Similar examples could be constructed with other values of $\alpha$. However, note that no state $\rho_{a,p}$ lying outside the separable ball could be detected 
by Eqs.~\eqref{MainResult:3a}-\eqref{MainResult:3b}.

\subsection{Non-invertible maps}

In this subsection we will consider (sufficient) separability criterion arising from some non-invertible maps. Due to non-invertibility of the maps involved, the derived criteria would not follow from the general theorem of this paper. It is quite challenging to generalize them to the invertible case.  

Our first example is about the map $\Phi_\alpha\colon\mathcal{B}\left(\CC^2\otimes\CC^N\right) \to \mathcal{B}\left(\CC^2\otimes\CC^N\right)$ defined as follows: any Hermitian $X \in \mathcal{B}\left(\CC^2\otimes\CC^N\right)$ may be written as
\[ X = \begin{pmatrix}
X_{11} & X_{12} \\ X^\dagger_{12} & X_{22}
\end{pmatrix},\]
where $X_{ij} \in\mathcal{B}\left(\CC^N\right)$. Let
\begin{equation}\label{}
  \Phi_\alpha(X) := \left( \begin{array}{cc} \mathbb{I}_N {\rm Tr}X_{11} & \alpha B  \\ \alpha B^\dagger & \mathbb{I}_N {\rm Tr} X_{22} \end{array} \right)
\end{equation}
with $ B= X_{12} + R_N(X_{12}^\dagger)$ and $R_N\colon \mathcal{B}\left(\CC^N\right)\rightarrow \mathcal{B}\left(\CC^N\right)$ being the reduction map. It was shown in Ref.~\cite{Chruscinski+2.PRA.2009} that $\Phi_{\alpha=1}$ defines an optimal positive map. Clearly, $\Phi_\alpha$ is positive for all $|\alpha| \leq 1$ (it is optimal for $|\alpha|=1$).

\begin{thm} \label{Thm:Sep.of.Phi.alpha}$\Phi_\alpha(|\psi\rangle\langle\psi|)$ is separable if $|\alpha| \leq 1$.
\end{thm}
\begin{proof}
The proof follows directly from Proposition~1 of Ref.~\cite{Gurvits+Barnum.PRA.2002},  since $\Phi_\alpha(|\psi\rangle\langle\psi|)$ is locally equivalent to
\[\begin{pmatrix}
\mathbb{I}_N & A \\ A^\dagger & \mathbb{I}_N
\end{pmatrix} \]
with $AA^\dagger \leq \mathbb{I}_N$.
\end{proof}

We now consider a slightly generalized map $\Psi_\alpha\colon\mathcal{B}\left(\CC^4\otimes\CC^N\right) \to \mathcal{B}\left(\CC^4\otimes\CC^N\right)$ defined as follows: any Hermitian $X \in \mathcal{B}\left(\CC^4\otimes\CC^N\right)$ may be written as
\[X = \begin{pmatrix}
X_{11} & X_{12} \\ X^\dagger_{12} & X_{22} 
\end{pmatrix},\]
where $X_{ij} \in \mathcal{B}\left(\CC^2\otimes\CC^N\right)$. Let
\begin{equation}\label{}
  \Psi_\alpha(X) := \left( \begin{array}{cc} \mathbb{I}_{2N} {\rm Tr}X_{11} & \alpha C  \\ \alpha C^\dagger & \mathbb{I}_{2N} {\rm Tr} X_{22} \end{array} \right)
\end{equation}
with $C= X_{12} + U \overline{X_{12}} U^\dagger$, $U$ being an arbitrary antisymmetric $2N \times 2N$ unitary matrix. It is known \cite{Chruscinski+Pytel.PRA.2010} that $\Psi_{\alpha=1}$ defines an optimal positive map. It is also clear that $\Psi_\alpha$ is positive for all $|\alpha| \leq 1$ (it is optimal for $|\alpha|=1$). 

\begin{thm} $\Psi_\alpha(|\psi\rangle\langle\psi|)$ is separable if $|\alpha| \leq 1$.
\end{thm}

The proof is essentially the same as that of Theorem~\ref{Thm:Sep.of.Phi.alpha}.

%=================================================== Sec-V ======================================================================%
\section{\label{Sec:MxN}Main Results for $\CC^M\otimes\CC^N$}

\subsection{\label{Subsection:M.N.Reduction-Breuer}Reduction- and Breuer-Hall-like maps}
There is a  general and simple generalization of the Theorem
\ref{thm:general}, summarized here as the Corollary.

\begin{cor}

\label{corr:lewenstein-witness} Let ${\bf p} =\{\alpha, \beta,
\gamma, \delta\}$ and $\Lambda_{\bf p}(\rho)= {\rm Tr}(\rho)\id +
\alpha \rho + \beta \tilde\rho_A + \gamma \tilde\rho_B  + \delta
\tilde\rho$ be a family of maps with the parameters fulfilling the
conditions such that $\rho\ge 0$ $\Rightarrow$ $\Lambda_{\bf p}(\rho) \in \Sigma$.
Then, if $w=P+Q^{T_A}$ is a decomposable pre-witness with
$P, Q\ge 0$, $\Lambda_{\bf p}(w)=:\sigma \in \Sigma$,
i.e. $\sigma$ is also separable.
\end{cor}
\begin{proof}
The proof follows easily from the fact that $\Lambda_{\bf p}(w^{T_A})=\Lambda_{\bf
p}(w)^{T_A}$. We present this corollary in this section since it holds in 
$\CC^M\otimes\CC^N$ equally well as in $\CC^2\otimes\CC^N$, provided appropriate 
generalizations of Breuer-Hall unitary operators are used.
\end{proof}

Let us now present our own proof of the Theorem~\ref{thm:barnum}, i.e., separability of
$\Lambda_\alpha(\rho)= {\rm Tr}(\rho)\id + \alpha \rho$ for arbitrary $M\otimes N$ $\rho$.

\begin{proof}
We discuss in detail the cases $M=3,4$ -- the generalizations to
arbitrary $M\le N$ is the straightforward. As before, it  is
enough to prove for pure states, $\rho=|\Psi\rangle\langle \Psi|$.
For $M=3$, $|\Psi\rangle$ has maximally Schmidt rank 3, and,
without loosing generality, we can assume that $|\Psi\rangle
=\lambda_0 |0\rangle\otimes |0\rangle + \lambda_1|1\rangle\otimes
|1\rangle + \lambda_2|2\rangle\otimes |2\rangle$. It is enough to
check then positivity and separability of
$\Lambda(|\Psi\rangle\langle \Psi|)$ on a $3\otimes 3$ space
spanned by $|0\rangle$, $|1\rangle$  and $|2\rangle$ in both
Alice's and Bob's spaces. This is achieved by observing first that
if $|\Psi\rangle$ is a normalized  state, than obviously positivity of
$\Lambda_\alpha(|\Psi\rangle\langle \Psi|)$ requires $\alpha\ge
-1$. Direct inspection shows that for $\alpha=-1$, $\Lambda_\alpha(|\Psi\rangle\langle \Psi|)$  is a sum of three 
positive matrices that effectively act on two-qubit Hilbert spaces and are given by 
%
%spanned by $|i\rangle$ and 
%$|j\rangle$ for Alice and Bob: 
%
\begin{equation}
\sigma_{ij}=\ket{\varphi_{ij}}\!\bra{\varphi_{ij}}+\ket{ij}\!\bra{ij}+\ket{ji}\!\bra{ji},
\end{equation}
where $\ket{\varphi_{ij}}=\lambda_j\ket{i}-\lambda_i\ket{j}$ with $i<j=0,1,2$.
%
%\begin{equation}
%\Lambda_\alpha(|\Psi\rangle\langle \Psi|)=
% \sum_{(i,j)}^2\begin{pmatrix}
%1-\lambda_i^2 & 0 & 0 & -\lambda_i\lambda_j \\
%0 & 1+ & 0 & 0 \\
%0 & 0 & 1 & 0  \\
%-\lambda_i\lambda_j & 0 & 0 &
%1-\lambda_j^2.
%\end{pmatrix},
%\end{equation}
%where $(i,j)$ denote disting (unordered) pairs of $i\ne j$. 
%
These matrices are PPT and thus separable. 

For $\alpha\ge 0$, we decompose $\Lambda_\alpha(|\Psi\rangle\langle \Psi|)
=\sigma^{T_A} + D$, where $\sigma^{T_A}$ is separable by construction,
and $D$ is positive diagonal in the computational product basis,
{\it ergo} separable by direct inspection. We consider a family of
product vectors
$|p(\phi,\psi)\rangle=A(1,ae^{i\phi},be^{i\psi})^{\otimes 2}$ and
the separable state
\[\sigma= \int d\phi/2\pi \int d\psi/2\pi |p(\phi,\psi)\rangle \langle
p(\phi,\psi)|.\] The parameters can be chosen indeed in such a way
that $\Lambda_\alpha(|\Psi\rangle\langle \Psi|) -\sigma^{T_A}$ is
diagonal. To this aim we choose $A^2=\alpha\lambda_0^2$,
$a^2=\lambda_1/\lambda_0$, $b^2=\lambda_2/\lambda_0$. Checking the
explicit conditions that $d$ is positive implies that
$1-A^2a^2=1-\alpha\lambda_0\lambda_1\ge 0$, i.e. $\alpha\le 2$
since $\lambda_0$ and $\lambda_1$ are the highest and the second
highest Schmidt coefficients, respectively. 
The case $-1<\alpha<0$ follows from convexity.

The proof for the case $M=4$ is analogous. We have now four
Schmidt coefficients, $\lambda_k$, $k=0, \ldots, 3$, and take
$|p(\phi,\psi)\rangle=A(1,ae^{i\phi},be^{i\psi},
ce^{i\theta})^{\otimes 2}$, set $A^2=\alpha\lambda_0^2$,
$a^2=\lambda_1/\lambda_0$, $b^2=\lambda_2/\lambda_0$, and
$c^2=\lambda_3/\lambda_0$. Generalization for $M>4$ follows the same pattern.
\end{proof}

The above theorem can be generalized to the sufficient conditions
for states with Schmidt number $n\le M$,  $\sigma_n\in \Sigma_n$.
Similar  results were obtained earlier  by L. Clarisse
\cite{Clarisse.JPA.2006}, but our proofs are different, so we present them
here.
\begin{thm}
\label{thm:clarisse} (see also \cite{Clarisse.JPA.2006}) Let $\Lambda_\alpha(\rho)=
{\rm Tr}(\rho)\id + \alpha \rho $ be the family of maps, and
$-1\le\alpha \le n+1$. Then $\rho\ge 0$ $\Rightarrow$
$\Lambda_\alpha(\rho)=\sigma_n\in \Sigma_n$, i.e. $\sigma_n$ has the
Schmidt number $\le n$.
\end{thm}
\begin{proof}
We discuss in detail the case $M=3$ and $n=2$ -- the
generalization to arbitrary $M\le N$ is the straightforward.  As
before it is enough to prove for pure states,
$\rho=|\Psi\rangle\langle \Psi|$. For $M=3$, $|\Psi\rangle$ has
maximally Schmidt rank 3, and, without loosing generality, we can
assume that $|\Psi\rangle =\lambda_0 |0\rangle\otimes |0\rangle +
\lambda_1|1\rangle\otimes |1\rangle + \lambda_2|2\rangle\otimes
|2\rangle$. It is enough to check then positivity and Schmidt number
of $\Lambda(|\Psi\rangle\langle \Psi|)$ on a $3\otimes 3$ space
spanned by $|0\rangle$, $|1\rangle$  and $|2\rangle$ in both
Alice's and Bob's spaces. This is achieved by observing first that
if $|\Psi\rangle$ is a normalized state, than obviously positivity of
$\Lambda_\alpha(|\Psi\rangle\langle \Psi|)$ requires $\alpha\ge
-1$. Moreover, we know that for $\alpha=-1$, $\Lambda_\alpha(|\Psi\rangle\langle \Psi|)$ is separable, {\it ergo} it obviously has the Schmidt number $1\le n=2$. For $\alpha\ge 0$, we decompose directly $\Lambda_\alpha(|\Psi\rangle\langle \Psi|)
=\sigma^{T_A}+ D_2^{T_A}$, where $\sigma$ is separable by construction, and $D_2$ has a form  that explicitly implies that
its Schmidt number is smaller equal 2.
As before, we consider a
family of product vectors
$|p(\phi,\psi)\rangle=A(1,ae^{i\phi},be^{i\psi})^{\otimes 2}$ and
the separable state
\[\sigma= \int d\phi/2\pi \int d\psi/2\pi |p(\phi,\psi)\rangle \langle
p(\phi,\psi)|.\] The parameters can be chosen indeed in such a way
that $\Lambda_\alpha(|\Psi\rangle\langle \Psi|) -\sigma^T_A$ has
only two off--diagonal elements, corresponding to the subspace
spanned by $|0\rangle$ and $|1\rangle$. To this aim we choose
$A^2=\epsilon\alpha\lambda_0^2$, $a^2=\lambda_1/\lambda_0$,
$b^2=\lambda_2/\lambda_0\epsilon$, with some $\epsilon \ge 0$.
Checking the explicit conditions that $d$ is positive implies,
among others,  that $1-A^2b^2=1-\alpha\lambda_0\lambda_2\ge 0$,
i.e. $\alpha\le 3$, since $\lambda_0$ and $\lambda_2$ are
the highest and the third highest Schmidt coefficients,
respectively. Other conditions, such as  $1-A^2a^2=1-\epsilon\alpha\lambda_0\lambda_1\ge 0$, or  $1+\alpha\lambda_2^2-A^2b^4=1+\alpha\lambda_2^2(1-1/\epsilon)\ge 0$, can be fulfilled with
$\epsilon=2/3$, for instance. 

Generalization for $M\ge 4$
and $n=2$ is  straightforward: we set $A^2=\epsilon\alpha\lambda_0^2$, $a^2=\lambda_1/\lambda_0$,
$b^2=\lambda_2/\lambda_0\epsilon$, $c^2=\lambda_3/\lambda_0\epsilon$, ...  with some $\epsilon =2/3$. For  $M\ge 4$
and $n=3$ we set  $A^2=\epsilon\alpha\lambda_0^2$, $a^2=\lambda_1/\lambda_0$,
$b^2=\lambda_2/\lambda_0$, $c^2=\lambda_3/\lambda_0\epsilon$, etc.  with  $\epsilon =1/2$. 
\end{proof}

Much stronger result than that  of Theorem \ref{thm:clarisse} was proven in Ref.~\cite{Clarisse.JPA.2006} for $M=N=d$.

\begin{thm}\label{Schmidt}
Let $\Lambda_\alpha(\rho)=
{\rm Tr}(\rho)\id + \alpha \rho $ be a family of maps with
$-1\le\alpha \le 2(dn-1)/(d-n)$. Then, if $\rho\ge 0$ $\Rightarrow$
$\sigma = \Lambda_\alpha(\rho)\in \Sigma_n$, i.e. $\sigma$ has the
Schmidt number $\le n$.
\end{thm}
\begin{proof}Let $\ket{\psi}$ be a two-qudit state with Schmidt coefficients $\lambda_0\geq\lambda_1\geq \ldots$. 
It was shown in Ref. \cite{Clarisse.JPA.2006} that the (unnormalized)
mixture of $\ket{\psi}$ and the maximally mixed state, 
%
%\begin{equation}
$\mathbbm{1}+\alpha\ket{\psi}\!\bra{\psi}$,
%\end{equation}
has Schmidt number $n$ if 
\begin{equation}
\alpha\leq \frac{dn-1}{(d-n)\lambda_0\lambda_1}.
\end{equation}
This implies that $\mathbbm{1}+\alpha\ket{\psi}\!\bra{\psi}$ has Schmidt number at most 
$n$ for any $\ket{\psi}$ if 
\begin{equation}
\alpha\leq \frac{2(dn-1)}{d-n}.
\end{equation}
%
%one can always construct a pure state $\ket{\psi}$ 
%of Schmidt rank $k$ for any $k$
%in which $\lambda_0=1/2$, $\lambda_1=\sqrt{1/2-\epsilon}$ and 
%the remaining Schmidt coefficients are $\sqrt{\epsilon/(k-1)}$.
\end{proof}
The value $2(dn-1)/(d-n)$ is larger than $n+1$ and it recovers the bound $\alpha\leq 2$ for separability. Still, the proof of Theorem \ref{thm:clarisse} is relatively simple and we expect that it  could be generalized to other situations, such as the multipartite case etc.  That is why we decided to present it here.

Let us finally notice that Theorem \ref{Schmidt} implies the following sufficient 
criterion for being an element of $\Sigma_n$, which generalizes Theorem \ref{cor:barnum1}. 
\begin{cor}If for some $\alpha\in[-1,2(dn-1)/(d-n)]$, 
$\Lambda^{-1}_\alpha(\sigma):= [\sigma -{\rm
Tr}(\sigma)\id/(2N+\alpha)]/\alpha \ge 0$, then 
$\sigma\in\Sigma_n$, i.e., the Schmidt number of $\rho$
is not larger than $n$.
\end{cor}

%\begin{thm}
%\label{cor:barnum1} {\bf (Sufficient separability criterion 1)} If
%$\Lambda^{-1}_\alpha(\sigma):= [\sigma -{\rm
%Tr}(\sigma)\id/(2N+\alpha)]/\alpha \ge 0$ then $\sigma\in \Sigma$,
%i.e. $\sigma$ is separable.
%\end{thm}

\subsection{\label{SubSec:Ando.arbitary.M.N}Ando-like maps for $M,N\ge3$}

Let us now consider the case $k=-1$ from the Section~\ref{Subsection:Andolike.2.N}, which can be easily generalized to the  
$\CC^M\otimes\CC^N$ case. Namely, for $\rho\in\mathcal{B}\left(\CC^M\otimes\CC^N\right)$, we define 
\[\epsilon(X):=\sum_{i=0,\,j=0}^{M-1,\,N-1}|i,j\rangle\langle i,j|X|i,j\rangle\langle i,j|,\]
 and 
\begin{equation}
\Lambda^{M\times N}(\rho):=(M-1)N\,\epsilon(\rho)+\sum_{m=0}^{N-1}\epsilon\left(S_B^m\rho S_B^{m\dagger}\right)
+\alpha\rho,
\end{equation}

where $S_B$ is the unitary shift operator in the Bob space.  

In this subsection we consider first the case of qutrits ($M=N=3$) and prove the following result.
\begin{thm}
The map
\begin{equation}\label{Sep.Criterion.3x3.Ando-type}
  \Lambda^{3\times 3}(\rho) = 6 \epsilon(\rho) + \sum_{m=0}^{2} \epsilon(S_B^m \rho S_B^{m\dagger}) + \alpha \rho
\end{equation}
is separable for $|\alpha|\leq 1/2$.
\end{thm}

\begin{proof}
We consider   $|\psi\rangle = \sum_{j=0}^{2}(\lambda_{0j} |0j\rangle + \lambda_{1j} |1j\rangle + \lambda_{2j} |2j\rangle$, with
$N_i= \sum_{j=0}^{2}|\lambda_{ij}|^2$ for $i=0,1,2$, and $\sum_{i=0}^2N_i=1$.  We obtain
\begin{widetext}
\begin{eqnarray*}
\Lambda^{3 \times 3}(\rho) &= &\left[ \begin{array}{ccc|ccc|ccc} 
6|\lambda_{00}|^2 + N_0& 0 & 0 & \alpha\lambda_{00}\lambda_{10}&  \alpha\lambda_{00}\lambda_{11} & \alpha\lambda_{00}\lambda_{12} & \alpha\lambda_{00}\lambda_{20}  &  \alpha\lambda_{00}\lambda_{21} &  \alpha\lambda_{00}\lambda_{22} \\ 
0 & 6|\lambda_{01}|^2 + N_0 & 0 & ... & ... & ... & ... & ... &\alpha\lambda_{01}\lambda_{22}  \\
0&  0 & 6|\lambda_{02}|^2 + N_0 & ... & ... & ... & ... & ... &  \alpha\lambda_{01}\lambda_{22} \\ \hline 
 \alpha\lambda_{00}^*\lambda_{10}^*& ... & ... &  6|\lambda_{10}|^2 + N_1   & 0 & 0 & ... & ... &  \alpha\lambda_{10}\lambda_{22} \\ 
\alpha\lambda_{00}^*\lambda_{11}^* & ... & ... & 0 &   6|\lambda_{11}|^2 + N_1    & 0 &... & ... &\alpha\lambda_{11}\lambda_{22} \\ 
\alpha\lambda_{00}^*\lambda_{12}^* & ... & ... &  0 &   0 & 6|\lambda_{12}|^2 + N_1 & ... &  &\alpha\lambda_{12}\lambda_{22}\\ \hline
\alpha\lambda_{00}^*\lambda_{20}^* & ... & ... & ... & ... & ... &  6|\lambda_{20}|^2 + N_2  & 0 & 0  \\
\alpha\lambda_{00}^*\lambda_{21}^* & ... & ... & ... & ...& ... & 0 &   6|\lambda_{21}|^2 + N_2 & 0  \\
\alpha\lambda_{00}^*\lambda_{22}^* & \alpha\lambda_{10}^*\lambda_{22}^* & \alpha\lambda_{02}^*\lambda_{22}^* & \alpha\lambda_{10}^*\lambda_{22}^* &  \alpha\lambda_{11}^*\lambda_{22}^* & \alpha\lambda_{12}^*\lambda_{22}^* & \alpha\lambda_{20}^*\lambda_{22}^* & \alpha\lambda_{21}^*\lambda_{22}^* &   6|\lambda_{22}|^2 + N_2\end{array} \right] .\\
&+& \sum_{i=0}^2|i\rangle\langle i|\otimes|f_i\rangle\langle f_i|, 
\end{eqnarray*}
\end{widetext}
where $|f_i\rangle=\sum_{j=0}^{2}\lambda_{ij}|j\rangle$. 
The proof follows exactly the same steps as in the previous section: we subsequently get rid of coherences in the above matrix by substracting matrices that live in $2\times 2$ space and are explicitely separable, provided $|\alpha|\le 1$. As before we end up with the conditions:
\begin{widetext}
\begin{equation*}
 \left[\begin{array}{ccc} (6-4|\alpha|)|\lambda_{i0}|^2  +(1-2|\alpha|)N_i   & 0 & 0  \\ 
0 &  (6-4|\alpha|)  |\lambda_{i1}|^2   +(1-2|\alpha|)N_i& 0  \\ 
0 &  0 &  (6-4|\alpha|)  |\lambda_{i2}|^2  +(1-2|\alpha|)N_i  \end{array} \right]+\alpha |f_i\rangle\langle f_i|  \ge 0,
\end{equation*}
\end{widetext}
Clearly, we must have $|\alpha|\le 1/2$ to assure non-negativity. Moreover, using again Ref.~\cite{Lewenstein+Sanpera.PRL.1998} we obtain that
\[ \sum_{j=0}^{2}\frac{|\alpha||\lambda_{0j}|^2}{ (6-4|\alpha|)|\lambda_{0j}|^2   +(1-2|\alpha|)N_0}\le 1,\] which implies 
$-9/13\le \alpha$, so that ultimately  the condition $|\alpha|\le 1/2$  is decisive. 
\end{proof}

Generalization to the case of arbitrary $M\times N$ yields 
\[ \sum_{j=0}^{N-1}\frac{|\alpha||\lambda_{0j}|^2}{ (M-1)(N-2|\alpha|)|\lambda_{0j}|^2   +(1-(M-1)|\alpha|)N_0}\le 1,\]
which leads to $-NM/(NM+2(M-1)\le \alpha$, which is also less restrictive that  $|\alpha|\le 1/2$, i.e. irrelevant.  
Therefore we have the following general result. 
\begin{thm}
If $\rho\ge 0$, then
\begin{equation}\label{}
  \Lambda^{M\times N}(\rho) = (M-1)N \epsilon(\rho) +  \sum_{m=0}^{N-1} \epsilon(S_B^m \rho S_B^{m\dagger}) + \alpha \rho
\end{equation}
is separable for $|\alpha|\leq 1/2$.
\end{thm}

Similar results can be proven analogously for the $k$-dependent $\Lambda$,  namely, for $k=0,1,\ldots$, we consider the following map
 \begin{align*}
 \Lambda^{M\times N}_k(\rho) &:= (N-1) \epsilon(\rho) +  \sum_{i=0}^{M-2}\sum_{j=0}^{N-1}  \epsilon(S_A^iS_B^j \rho S_A^{i\dagger} S_B^{j\dagger}) \\ 
 &+ \sum_{j=0}^{k}  \epsilon(S_A^{M-1}S_B^j \rho S_A^{M-1 \dagger} S_B^{j\dagger})  + \alpha \rho\numberthis\label{Def:LambdaMNK.alpha}.
 \end{align*}

\begin{thm}
If $\rho\ge 0$, then $\Lambda^{M\times N}_0(\rho)$, as defined in Eq.~\eqref{Def:LambdaMNK.alpha},
is separable for $\max[-(2N-1)/(N+(N+2)(M-1)), -1/2]\leq \alpha\leq 1/2$.
\end{thm}

\begin{proof}
As above we first present the proof  for $M=N=3$ in more detail. The starting point is to write down the explicit matrix from of $\Lambda^{3 \times 3}(\rho)$.
Performing the same steps as before and assuming that $|\alpha|\le 1$ we end up with the condition for separability
\begin{align*}\diag&\left\{(2-4|\alpha|)|\lambda_{00}|^2  +(1-2|\alpha|)N_0  +N_1 +|\lambda_{20}|^2,\right.\\
 &(2-4|\alpha|)|\lambda_{01}|^2  +(1-2|\alpha|)N_0  +N_1 +|\lambda_{21}|^2,\\
 &\left.(2-4|\alpha|)|\lambda_{02}|^2  +(1-2|\alpha|)N_0  +N_1 +|\lambda_{22}|^2\right\}\\
&+\alpha |f_i\rangle\langle f_i|  \ge 0,
\end{align*}
and analogous conditions  for $i=1,2$. Clearly, $|\alpha|\le 1/2$ is necessary; a conservative bound gives $-5/13\le \alpha$.
Generalization to arbitrary $M,N$ leads to 
\[\sum_{j=0}^{N-1}\frac{|\alpha||\lambda_{0j}|^2}{ (N-1-2(M-1)|\alpha|)|\lambda_{0j}|^2   +(1-(M-1)|\alpha|)N_0}\le 1,\]
implying ${\rm max}[-(2N-1)/(N+(N+2)(M-1)),-1/2]\le \alpha$.
\end{proof}

\subsection{Examples} Similar to Section~\ref{Subsection:examples.2.N}, we will present some examples of states which lie outside the separable ball, but nonetheless are detected to be separable by our simplest criterion in Eq.~\eqref{Sep.Criterion.3x3.Ando-type}.

Let us consider the following class of $3\otimes 3$ states from Ref.~\cite{AllHorodeckis.2001},
\begin{align*}
 \rho_\beta&=\frac{2}{7}|\Phi\ran\lan\Phi|+\frac{\beta}{7}\sigma_+ +\frac{5-\beta}{7}\sigma_-,\quad \beta\in[0,5],\\
\text{where~~ }|\Phi\ran&= \frac{1}{\sqrt{3}}\sum_{k=0}^2|kk\ran,\\
 \sigma_+&=\frac{1}{3}\left( |01\ran\lan01|+|12\ran\lan12|+|20\ran\lan20|\right),\\
 \sigma_-&=\frac{1}{3}\left( |10\ran\lan10|+|21\ran\lan21|+|02\ran\lan02|\right).
\end{align*}
This is an interesting class of states: for $0\leq\beta<1$ and $4<\beta\le 5$ it is NPT, for  $2\le\beta\le 3$ it is separable, and for $3<\beta\le 4$ it is bound entangled.  One verifies that for
$0\leq \beta <\left(110-\sqrt{4495}\right)/44$, or  $\left(110+\sqrt{4495}\right)/44<\beta \leq 5$, 
$\rho_{\beta}$ is NPT (hence entangled), but for all $\alpha\in[-1/2,\,1/2]$, $\Lambda^{3\times 3}\left(\rho_{\beta}\right)$ lies outside the separable ball yet detected to be separable by Eq.~\eqref{Sep.Criterion.3x3.Ando-type}. Similar examples can be constructed for the other maps as well.

%%%%%%%%%%%%%%%%%%%%%%%%%% Sec-VI %%%%%%%%%%%%%%%%%%%%%%%%%%%%%%%%%%%%%%%%%%%%%%%%%%%

\section{Conclusions and Outlook}

We have presented in this paper several families of linear maps, that have  a property that when applied to states, result in separable states, or states from a certain class of entanglement. When invertible, such maps allow to derive sufficient crtiteria for separability, or for a certain class of entanglement.  Such criteria have  rather obvious possible applications and implications in quantum infomation
science: they can namely be applied to determine whether and how much entanglement is needed for various specific protocols for quantum computation or other quantum task, similarly as it was discussed in Ref.~\cite {Braunstein+5.PRL.1999}. 
We have formulated also several open questions in particular concerning certain kind of Breuer-Hall-like and reduction-like   maps acting in
$\mathcal{B}\left(\CC^M\otimes\CC^N\right)$, Ando-like maps acting in $\mathcal{B}\left(\CC^M\otimes\CC^N\right)$, and certain non-invertible maps in $\mathcal{B}\left(\CC^2\otimes\CC^N\right)$ and $\mathcal{B}\left(\CC^4\otimes\CC^N\right)$. It is quite remarkable that even our simplest criteria obtained for the case of reduction-like maps is independent to some well-known criteria from the literature. We have also presented different proof for already known results on separability of some class of states, and hopefully the technique is applicable for other, so far not-considered-in-literature  situations, such as the case of many parties, symmetric states etc. 

\label{Conclusion}

\begin{acknowledgments}
ML, RA, and SR acknowledge financial support from the John Templeton Foundation, the EU grants 
OSYRIS (ERC-2013-AdG Grant No. 339106), QUIC (H2020-FETPROACT-2014 No. 641122),
and SIQS (FP7-ICT-2011-9 No. 600645), the Spanish MINECO grants FOQUS (FIS2013-
46768-P) and "Severo Ochoa" Programme (SEV-2015-0522), and the Generalitat de Catalunya
grant 2014 SGR 874.
\end{acknowledgments}

\end{document}